\documentclass{article}
\usepackage{fullpage}
\usepackage{graphicx}
\usepackage{appendix}
\usepackage{amsmath,amssymb,amsthm,bm,fancybox}
\usepackage[dvipdfmx]{color}
\usepackage{cite}

\makeatletter
\newcommand*{\samethanks}[1][\value{footnote}]{\footnotemark[#1]}
\makeatother

\theoremstyle{definition}
\newtheorem{theorem}{Theorem}[section]
\newtheorem{definition}[theorem]{Definition}

\newtheorem{corollary}[theorem]{Corollary}
\newtheorem{lemma}[theorem]{Lemma}
\newtheorem{remark}[theorem]{Remark}

\newcommand{\bin}{\{0,1\}}
\newcommand{\ra}{\rightarrow}

\newcommand{\ol}[1]{\overline{#1}}

\newcommand{\Z}{\mathbb{Z}}
\newcommand{\wt}{\mathsf{wt}}
\newcommand{\cequal}{\sim}

\newcommand{\insto}[1]{\overset{{#1}}{\to}}

\newlength{\ralen}
\def\club{\clubsuit}
\def\heart{\heartsuit}
\newcommand{\red}{\raisebox{\ralen}{\framebox(11,13){${\scriptstyle \heart}$}}}
\newcommand{\blk}{\raisebox{\ralen}{\framebox(11,13){${\scriptstyle \club}$}}}
\newcommand{\back}{\raisebox{\ralen}{\framebox(11,13){\bf ?}}}

\title{Cyclic Equalizability of Words and Its Application to Card-Based Cryptography}
\author{Kazumasa Shinagawa\thanks{University of Tsukuba, Ibaraki, Japan} \thanks{National Institute of Advanced Industrial Science and Technology, Tokyo, Japan} \and Koji Nuida\thanks{Institute of Mathematics for Industry, Kyushu University, Fukuoka, Japan} \samethanks[2]}

\begin{document}

\maketitle              
\begin{abstract}
Card-based cryptography is a research area to implement cryptographic procedures using a deck of physical cards. In recent years, it has been found to be related to finite group theory and algebraic combinatorics, and is becoming more and more closely connected to the field of mathematics. In this paper, we discuss the relationship between card-based cryptography and combinatorics on words for the first time. In particular, we focus on cyclic equality of words. We say that a set of words are cyclically equalizable if they can be transformed to be cyclically equal by repeated simultaneous insertion of letters. The main result of this paper is to show that two binary words of equal length and equal Hamming weight are cyclically equalizable. As applications
of cyclic equalizability to card-based cryptography, we describe its applications to the information erasure problem and to single-cut full-open protocols. 
\end{abstract}
\section{Introduction}

\subsection{Background}

Suppose that $n$ players $P_1, P_2, \ldots, P_n$ each having inputs $x_1, x_2, \ldots, x_n \in \bin$, respectively, wish to compute the output value $f(x_1, x_2, \ldots, x_n) \in \bin$ for some function $f\colon \bin^n \ra \bin$. 
This is called distributed computing. 

\emph{Secure computation} is a kind of distributed computing where the inputs are not leaked from the protocol execution more than what is computed from the output value. 
It has many practical applications, e.g., a voting for two candidates: given a candidate 0 and a candidate 1, each player inputs one bit as a vote, and the output value is the index of the candidate who receives a majority of votes. 
Although secure computation is usually implemented on electronic computers, there is a line of research that implements secure computation using a deck of physical cards, and such a research area is called \emph{card-based cryptography}~\cite{BoerEC1989, KilianC1994,MizukiIJIS2014}. 

In card-based cryptography, we usually use a deck of cards whose faces are either $\blk$ or $\red$ and whose backs are the same symbol $\back\,$. 
Here, the cards with the same symbol are assumed to be indistinguishable.
We use a pair of cards to encode one bit information by $\blk\,\red = 0$ and $\red\,\blk = 1$, and a pair of face-down cards $\back\,\back$ encoding $x \in \bin$ is called a \emph{commitment} to $x$.
Given $n$ commitments to $x_1, x_2, \ldots, x_n \in \bin$, a protocol for $f: \bin^n \ra \bin$ computes the output value $f(x_1, x_2, \ldots, x_n)$ without revealing the inputs beyond the output by performing a sequence of physical operations such as shuffling. 

The history of card-based cryptography began with the five-card trick proposed by den Boer~\cite{BoerEC1989} in 1989. 
The purpose of den Boer's paper itself was to construct a cryptographic protocol, which is assumed to be implemented on electronic computers, and physical cards were used to illustrate the idea of the protocol. 
In 1994, Cr\'{e}peau and Kilian~\cite{KilianC1994} showed the feasibility of card-based protocol for any function and developed the basic theory of card-based cryptography. 
The turning point in card-based cryptography was the proposal of a rigorous computational model of card-based protocols known as the Mizuki--Shizuya model by Mizuki and Shizuya~\cite{MizukiIJIS2014}. 
After Mizuki--Shizuya's paper, this research field was revitalized as it became actively studied both in terms of protocol construction (e.g.,~\cite{ShinagawaDAM2021}) and impossibility proofs regarding the number of cards and the number of operations (e.g.,~\cite{KochAC2015,KastnerAC2017}). 
In recent years, the connection with mathematical fields has been deepening such as finite group theory~\cite{KanaiCommunAlgebra2024,ShinagawaFUN2024}, algebraic combinatorics~\cite{SugaICCE2022}, and formal method~\cite{KochAC2019}. 
However, such an attempt has only just begun, and the development of new connections with other mathematical fields is important for the further development of this field.

In this paper, we focus on combinatorics on words~\cite{Lothaire1983,Lothaire2002,Lothaire2005} and study a new notion of \emph{cyclic equalizability} of words.
As discussed in Section~\ref{ss:combinatorics_from_card}, this concept appears naturally in card-based cryptography, but has never been studied. 
This is the first study to discuss card-based cryptography and combinatorics on words, and our main contribution is to bring these two fields together.
We believe that this paper will lead to the development of the relation between card-based cryptography and combinatorics on words. 

\subsection{Cyclic Equalizability from Card-Based Cryptography}\label{ss:combinatorics_from_card}

In this section, we see that cyclic equalizability appears naturally in card-based cryptography by observing a protocol called the \emph{five-card trick}~\cite{BoerEC1989}, which uses a \emph{random cut}. 

A random cut is a shuffling operation that cyclically shifts a sequence of 
face-down
cards a random number of times. For a sequence of 
face-down
cards $(c_0, c_1, \ldots, c_{n-1})$, we obtain $(c_r, c_{r+1}, \ldots, c_{r+n-1})$, where $r \in \Z_n$ is a secret number chosen uniformly at random and the subscript is modulo $n$. 
Here, the secret number $r$ is hidden from all players and nobody can guess it. 
We use $\langle \cdot \rangle$ to denote the application of a random cut to a sequence of cards. 
Random cuts are considered to be the simplest shuffling operation proposed so far, since they can be easily performed by a well-known shuffling method called \emph{Hindu shuffle}.

The five-card trick is a protocol for the AND function $a\wedge b$ using a random cut. 
A secure computation protocol for the AND function can be used in the following situation:
\begin{quote}
{\it Alice and Bob are going to decide if they want to go on a date this weekend. 
However, Alice and Bob are very shy and embarrassed to tell each other directly whether they want to go on a date or not. 
This can be especially awkward if one of them wants to go on a date and the other does not.
Is there any way to know just whether they both want to go on a date or not?}
\end{quote}
If we assume that \lq\lq 1\rq\rq{} means \lq\lq I want to go on a date\rq\rq{} and \lq\lq 0\rq\rq{} means \lq\lq I don't want to go on a date\rq\rq, and Alice and Bob each have inputs $a,b\in \bin$, then this problem can be solved by secure computation for the AND function $a\wedge b$. 
The \emph{five-card trick} computing the AND function proceeds as follows. 

\begin{enumerate}
\item Make the input commitments to $a, b \in \bin$. 
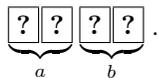
\[
\underbrace{\back\,\back}_{a} \, \underbrace{\back\,\back}_{b}\,.
\]
\item Swap the left and right cards of the commitment to $a$.
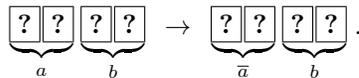
\[
\underbrace{\back\,\back}_{a} \, \underbrace{\back\,\back}_{b}
~\ra~
\underbrace{\back\,\back}_{\ol{a}} \, \underbrace{\back\,\back}_{b}\,.
\]
\item Insert a face-down card $\red$ in the center.
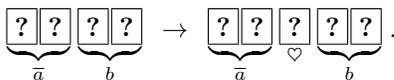
\[
\underbrace{\back\,\back}_{\ol{a}} \, \underbrace{\back\,\back}_{b}
~\ra~
\underbrace{\back\,\back}_{\ol{a}} \, \underset{\heart}{\back} \, \underbrace{\back\,\back}_{b}\,.
\]
\item Apply a random cut. 
\[
\Bigl\langle\;
\overset{1}{\back}\,\overset{2}{\back}\,\overset{3}{\back}\,\overset{4}{\back}\,\overset{5}{\back}
\;\Bigr\rangle
~\ra~
\begin{cases}
\overset{1}{\back}\,\overset{2}{\back}\,\overset{3}{\back}\,\overset{4}{\back}\,\overset{5}{\back} & \mbox{w.p. $1/5$;}\\
\overset{2}{\back}\,\overset{3}{\back}\,\overset{4}{\back}\,\overset{5}{\back} \, \overset{1}{\back} & \mbox{w.p. $1/5$;}\\
\overset{3}{\back}\,\overset{4}{\back}\,\overset{5}{\back}\, \overset{1}{\back}\,\overset{2}{\back} & \mbox{w.p. $1/5$;}\\
\overset{4}{\back}\,\overset{5}{\back} \, \overset{1}{\back}\,\overset{2}{\back}\,\overset{3}{\back}& \mbox{w.p. $1/5$;}\\
\overset{5}{\back} \, \overset{1}{\back}\,\overset{2}{\back}\,\overset{3}{\back}\,\overset{4}{\back}& \mbox{w.p. $1/5$.}
\end{cases}
\]
\item Open all cards. Output $0$ if it is the left case and $1$ if it is the right case. 
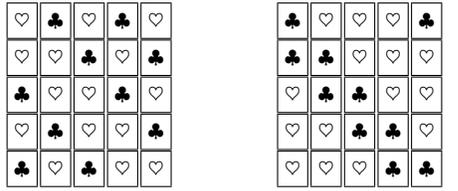
\[
\begin{tabular}{ccc}
$\red\,\blk\,\red\,\blk\,\red$ &~~~~~ & $\blk\,\red\,\red\,\red\,\blk$ \\
$\red\,\red\,\blk\,\red\,\blk$ &~~~~~ & $\blk\,\blk\,\red\,\red\,\red$ \\
$\blk\,\red\,\red\,\blk\,\red$ &~~~~~ & $\red\,\blk\,\blk\,\red\,\red$ \\
$\red\,\blk\,\red\,\red\,\blk$ &~~~~~ & $\red\,\red\,\blk\,\blk\,\red$ \\
$\blk\,\red\,\blk\,\red\,\red$ &~~~~~ & $\red\,\red\,\red\,\blk\,\blk$ 
\end{tabular}.
\]
\end{enumerate}
The security of the five-card trick follows from the fact that for any given input $(a,b)$ such that $a\wedge b = 0$, the resulting sequence is chosen uniformly at random from five patterns on the left side due to the effect of the random cut. 
Indeed, if $a\wedge b = 0$, a sequence of cards just after Step 2 is given as follows:
\[
\begin{tabular}{cc}
$\red\,\blk\,\blk\,\red$ & $(a,b) = (0,0)$; \\
$\red\,\blk\,\red\,\blk$ & $(a,b) = (0,1)$; \\
$\blk\,\red\,\blk\,\red$ & $(a,b) = (1,0)$.
\end{tabular}
\]
Then, a sequence of cards just after Step 3 is given as follows:
\[
\begin{tabular}{cc}
$\red\,\blk\,\red\,\blk\,\red$ & $(a,b) = (0,0)$; \\
$\red\,\blk\,\red\,\red\,\blk$ & $(a,b) = (0,1)$; \\
$\blk\,\red\,\red\,\blk\,\red$ & $(a,b) = (1,0)$. 
\end{tabular}
\]
Since these sequences are cyclically equal (see Definition~\ref{def:cyclically_equal}), they become indistinguishable after a random cut is applied. 
We get three cyclically equal sequences by inserting $\red$ in the middle. 
This is the key to understanding the structure of the five-card trick.

Let us reconsider the above as a problem of combinatorics on words.
From now on, we consider words on $\Sigma$ for a finite set of letters $\Sigma$. 
We say that $k$ words $w_1, w_2, \ldots, w_k \in \Sigma^n$ are \emph{cyclically equalizable} if all words can be transformed into cyclically equal words by repeated simultaneous insertion of letters, i.e., by inserting the same letters at the same positions (see Definition~\ref{def:cyclically-equalizable}).
The five-card trick corresponds to three binary words $1001, 1010, 0101$ with $\clubsuit = 0$ and $\heartsuit = 1$, which are cyclically equalizable by inserting $1$ in the middle. 
In this paper, we consider a question of what is the necessary and sufficient condition for the cyclic equalizability.
This is a problem of combinatorics on words that arises naturally from card-based cryptography. 

\subsection{Our Contribution}

In this paper, we show that two binary words $w_1, w_2 \in \bin^n$ are cyclically equalizable if and only if the Hamming weights of $w_1$ and $w_2$ are equal (Theorem~\ref{thm:criterion_for_two_binary_words}). 
To show the main theorem, we first show that the cyclic equalizability does not change before and after simultaneous insertion (Theorem~\ref{thm:equalizability_is_invariant_under_insertion}). 
Thanks to this, in the proof of Theorem~\ref{thm:criterion_for_two_binary_words}, we can assume that $w_1$ and $w_2$ have no common letters on the same position. 
This makes the discussion easy to understand.

As applications
of cyclic equalizability to card-based cryptography, we describe its applications to the information erasure problem and single-cut full-open protocols in Section~\ref{sec:application}. 
In particular, we discuss the fact that cyclic equalizability is equivalent to the possibility of information erasure (Section~\ref{ss:information_erasure}), and we show a lower bound on the number of cards in single-cut full-open protocols (Section~\ref{ss:scfo}). 

\section{Basic Notations}

Let $\Sigma$ be a finite set of \emph{letters}. 
For a subset $\Delta \subseteq \Sigma$, we denote by $\Delta^*$ the set of all finite sequences of letters in $\Delta$. 
Its elements $w \in \Delta^*$ are called \emph{words} over $\Delta$. 
The number of letters in $w$ is called the \emph{length} of $w$. 
For a letter $a \in \Sigma$ and a non-negative integer $\mu$, a word $aa\cdots a$ that repeats $a$ by $\mu$ times is denoted by $a^{\mu}$. 
For a positive integer $n$ and an integer $m$, the remainder of $m$ divided by $n$ (taken to be between $0$ and $n-1$) is denoted by $m \bmod n$.
For a special case, words over $\bin$ are called \emph{binary words}. 
For a binary word $w \in \{0,1\}^*$, the number of $1$'s that appear in $w$ is called the \emph{Hamming weight} of $w$ and is denoted by $\wt(w)$.

We identify $k$ words of equal length $w_i = a_{i,0} a_{i,1} \cdots a_{i,n-1} \in \Sigma^n$ ($i = 1,\dots,k$) with a $k \times n$ matrix formed by placing the word $w_i$ in the $i$-th row as follows:
\[
\begin{pmatrix}
w_1 \\
w_2 \\
\vdots \\
w_k \\
\end{pmatrix}
= 
\begin{pmatrix}
a_{1,0} & a_{1,1} & a_{1,2} & \cdots & a_{1,n-1} \\
a_{2,0} & a_{2,1} & a_{2,2} & \cdots & a_{2,n-1} \\
\vdots & \vdots & \vdots & \ddots & \vdots \\
a_{k,0} & a_{k,1} & a_{k,2} & \cdots & a_{k,n-1} \\
\end{pmatrix}.
\]
We also denote a column vector consisting of $k$ letters $c_1,\dots,c_k$ from the top by $(c_1,\dots,c_k)^{\mathrm{T}}$ or simply $(c_1 \cdots c_k)^{\mathrm{T}}$ without commas.

\begin{definition}[Cyclically Equal]
    \label{def:cyclically_equal}
    Two words $w_i = a_{i,0} a_{i,1} \cdots a_{i,n-1} \in \Sigma^n$ ($i = 1,2$) of equal length are said to be \emph{cyclically equal} if there exists an integer $\delta$ such that $a_{2,j} = a_{1,j + \delta \bmod n}$ for all $0 \leq j \leq n-1$.
    We denote this relation by $w_1 \cequal w_2$.
    (Note that $\cequal$ is an equivalence relation on $\Sigma^n$.)
    Similarly, $k$ words of equal length $w_i \in \Sigma^n$ ($i = 1,2,\dots,k$) are said to be \emph{cyclically equal} if any pair of words $w_i, w_j$ are cyclically equal.
\end{definition}

\begin{definition}[$\Delta$-Insertion/Deletion]
    \label{def:Delta-insertion}
    Let $\Delta \subseteq \Sigma$.
    Let $w_i \in \Sigma^n$ ($1 \leq i \leq k$) be $k$ words of length $n$. 
    For $w_1, w_2, \dots, w_k$, a \emph{$\Delta$-insertion} is an operation that transforms each $w_i = a_{i,0} a_{i,1} \cdots a_{i,n-1}$ into a word $u_0 a_{i,0} u_1 a_{i,1} \cdots u_{n-1} a_{i,n-1} u_n$, where $u_0,u_1,\dots,u_n \in \Delta^*$ are $n+1$ words over $\Delta$. 
    We write $(w_1,\dots,w_k)\insto{\Delta} (w'_1,\dots,w'_k)$ if $(w'_1,\dots,w'_k)$ is obtained from $(w_1,\dots,w_k)$ by a $\Delta$-insertion. 
    The reverse operation of a $\Delta$-insertion is called a \emph{$\Delta$-deletion}. 
\end{definition}

We simply write \lq\lq $c$-insertion\rq\rq{} instead of {}\lq\lq $\{c\}$-insertion\rq\rq{}.
Note that the lengths of the words obtained by a $\Delta$-insertion or a $\Delta$-deletion are all equal.

\begin{remark}
A $\Delta$-insertion means repetition of an insertion of a column vector $(cc \cdots c)^{\mathrm{T}}$ ($c \in \Delta$) in the matrix representation of $w_1, w_2, \dots, w_k$ as follows:
\[
\begin{pmatrix}
a_{1,0} & \cdots & a_{1,i} & a_{1,i+1} & \cdots & a_{1,n-1} \\
a_{2,0} & \cdots & a_{2,i} & a_{2,i+1} & \cdots & a_{2,n-1} \\
\vdots & \vdots & \vdots & \vdots & \ddots & \vdots \\
a_{k,0} & \cdots & a_{k,i} & a_{k,i+1} & \cdots & a_{k,n-1} \\
\end{pmatrix}
~\rightarrow~
\begin{pmatrix}
a_{1,0} & \cdots & a_{1,i} & c & a_{1,i+1} & \cdots & a_{1,n-1} \\
a_{2,0} & \cdots & a_{2,i} & c & a_{2,i+1} & \cdots & a_{2,n-1} \\
\vdots & \vdots & \vdots & \vdots & \vdots & \ddots & \vdots \\
a_{k,0} & \cdots & a_{k,i} & c & a_{k,i+1} & \cdots & a_{k,n-1} \\
\end{pmatrix}.
\]
Conversely, a $\Delta$-deletion means repetition of a deletion of a column vector $(cc \cdots c)^{\mathrm{T}}$ ($c \in \Delta$) in the matrix representation of $w_1, w_2, \dots, w_k$ as follows:
\[
\begin{pmatrix}
a_{1,0} & \cdots & a_{1,i} & c & a_{1,i+1} & \cdots & a_{1,n-1} \\
a_{2,0} & \cdots & a_{2,i} & c & a_{2,i+1} & \cdots & a_{2,n-1} \\
\vdots & \vdots & \vdots & \vdots & \vdots & \ddots & \vdots \\
a_{k,0} & \cdots & a_{k,i} & c & a_{k,i+1} & \cdots & a_{k,n-1} \\
\end{pmatrix}
~\rightarrow~
\begin{pmatrix}
a_{1,0} & \cdots & a_{1,i} & a_{1,i+1} & \cdots & a_{1,n-1} \\
a_{2,0} & \cdots & a_{2,i} & a_{2,i+1} & \cdots & a_{2,n-1} \\
\vdots & \vdots & \vdots & \vdots & \ddots & \vdots \\
a_{k,0} & \cdots & a_{k,i} & a_{k,i+1} & \cdots & a_{k,n-1} \\
\end{pmatrix}.
\]
\end{remark}

\begin{definition}[Cyclically Equalizable]
    \label{def:cyclically-equalizable}
    Let $\Delta \subseteq \Sigma$.
    We say that $k$ words $w_i \in \Sigma^n$ ($1 \leq i \leq k$) are \emph{$\Delta$-cyclically equalizable} if there exist $w'_1,\dots,w'_k \in \Sigma^*$ that are cyclically equal such that $(w_1,\dots,w_k)\insto{\Delta} (w'_1,\dots,w'_k)$. 
    When $\Delta = \Sigma$, we simply say that they are cyclically equalizable. 
\end{definition}

\begin{remark}
    \label{rem:initial_subword_is_not_needed}
    In Definition~\ref{def:Delta-insertion}, we can assume without loss of generality that $u_0$ is an empty word. 
    This is because whether $u_0 a_{i,0} u_1 a_{i,1} \cdots u_{n-1} a_{i,n-1} u_n$ ($1 \leq i \leq k$) are cyclically equal is not changed by moving $u_0$ to the rightmost. 
\end{remark}

Given $k$ words $w_i \in \Sigma^n$ ($1 \leq i \leq k$), the \emph{$\Delta$-cyclic equalizability problem} asks whether they are $\Delta$-cyclically equalizable or not, and how to insert if it is possible. 
In Section~\ref{ss:information_erasure}, we will see that the cyclic equalizability problem is equivalent to the \emph{information erasure problem} in card-based cryptography as an application of cyclic equalizability. 

\section{$\Delta$-Insertion Preserves Cyclic Equalizability}

We prove the following theorem.

\begin{theorem}
    \label{thm:equalizability_is_invariant_under_insertion}
    Let $\Delta \subseteq \Sigma$. 
    Suppose that we have $(w_1,\dots,w_k) \insto{\Delta} (w'_1,\dots,w'_k)$ for $2k$ words $w_i, w'_i \in \Sigma^*$ ($1 \leq i  \leq k$). 
    Then the following statements are equivalent.
    \begin{enumerate}
        \item $w_1,\dots,w_k$ are $\Delta$-cyclically equalizable.
        \item $w'_1,\dots,w'_k$ are $\Delta$-cyclically equalizable.
    \end{enumerate}
\end{theorem}

The following lemma is given for the proof of Theorem~\ref{thm:equalizability_is_invariant_under_insertion}.

\begin{lemma}
    \label{lem:inserting_equal_number_of_zeroes}
    Suppose that $w_i = a_{i,0} a_{i,1} \cdots a_{i,n-1} \in \Sigma^n$ ($1 \leq i \leq k$) are cyclically equal. 
    Then for any $c \in \Sigma$ and $\mu \geq 0$, $\widetilde{w}_i = a_{i,0} c^{\mu} a_{i,1} c^{\mu} \cdots a_{i,n-1} c^{\mu} \in \Sigma^{(\mu+1)n}$ ($1 \leq i \leq k$) are also cyclically equal. 
\end{lemma}

\begin{proof}
    It is sufficient to show the case where $k = 2$.
    From the assumption, there exists an integer $\delta$ such that $a_{2,j} = a_{1,j + \delta \bmod n}$ for all $0 \leq j \leq n-1$. 
    Let $\widetilde{w}_i = b_{i,0} b_{i,1} \cdots b_{i, (\mu+1) n - 1}$ ($b_{i,j} \in \Sigma$). 
    Then $b_{i,j} = a_{i,j/(\mu+1)}$ if $j$ is a multiple of $\mu+1$ and $b_{i,j} = c$ otherwise.
    Thus, for each $0 \leq j \leq (\mu+1) n - 1$, $b_{2,j} = b_{1,j + (\mu+1) \delta \bmod (\mu+1)n}$.
    Therefore $\widetilde{w}_i$ are also cyclically equal.
\end{proof}

\begin{proof}
    [\textbf{Theorem~\ref{thm:equalizability_is_invariant_under_insertion}}]
    Since it is trivial that the latter condition implies the former, we show the converse direction. 
    It is sufficient to show that the latter condition follows from the former condition when $(w_1,\dots,w_k)\insto{c} (w'_1,\dots,w'_k)$ for some letter $c \in \Delta$ because any $\Delta$-insertion can be decomposed into a series of such insertions. 

    Let $w_i = a_{i,0} a_{i,1} \cdots a_{i,n-1}$ ($a_{i,j} \in \Sigma$). 
    From Remark~\ref{rem:initial_subword_is_not_needed}, we can assume that $w'_i = a_{i,0} c^{\mu_0} a_{i,1} c^{\mu_1} \cdots a_{i,n-1} c^{\mu_{n-1}}$ for some common non-negative integers $\mu_0,\mu_1,\dots,\mu_{n-1}$. 
    Let $\nu := \max\{\mu_0,\mu_1,\dots,\mu_{n-1}\}$ and $\widehat{b} := b c^{\nu} \in \Sigma^{\nu+1}$ for each $b \in \Sigma$. 
    Then for $(w'_1,\dots,w'_k)$, a $c$-insertion yields $(u_1, \dots, u_k)$ for $u_i = \widehat{a_{i,0}} \widehat{a_{i,1}} \cdots \widehat{a_{i,n-1}} \in \Sigma^{(\nu+1)n}$. 
    From the assumption that $(w_1,\dots,w_k)$ are $\Delta$-cyclically equalizable and Remark~\ref{rem:initial_subword_is_not_needed}, we obtain 
    some common letters $d_{0,1},\dots,d_{0,\ell_0},\allowbreak\dots,d_{n-1,1},\dots,d_{n-1,\ell_{n-1}} \in \Delta$ such that 
    \[
    a_{i,0} d_{0,1} \cdots d_{0,\ell_0}~a_{i,1} d_{1,1} \cdots d_{1,\ell_1}~\cdots~a_{i,n-1} d_{n-1,1} \cdots d_{n-1,\ell_{n-1}}
    \]
    are cyclically equal. 
    Then, from Lemma~\ref{lem:inserting_equal_number_of_zeroes}, 
    \[
    \widetilde{u}_i 
    := \widehat{a_{i,0}} \widehat{d_{0,1}} \cdots \widehat{d_{0,\ell_0}}~\widehat{a_{i,1}} \widehat{d_{1,1}}\cdots \widehat{d_{1,\ell_1}}~ \cdots ~\widehat{a_{i,n-1}} \widehat{d_{n-1,1}} \cdots \widehat{d_{n-1,\ell_{n-1}}}
    \]
    are also cyclically equal. 
    Moreover, since $\widehat{d_{h,j}} \in \Delta^{\nu+1}$, we have 
    $(u_1,\dots,u_k) \insto{\Delta} (\widetilde{u}_1,\dots,\widetilde{u}_k)$. 
    As mentioned above, we have $(w'_1,\dots,w'_k) \insto{\Delta} (u_1,\dots,u_k)$, thus we also have 
    $(w'_1,\dots,w'_k) \insto{\Delta} (\widetilde{u}_1,\dots,\widetilde{u}_k)$, 
    therefore $(w'_1,\dots,w'_k)$ are $\Delta$-cyclically equalizable.
\end{proof}

The following corollary, which follows directly from Theorem~\ref{thm:equalizability_is_invariant_under_insertion}, is useful to reduce the $\Delta$-cyclic equalizability for a set of words to that of shorter words. 

\begin{corollary}
    \label{cor:equalizability_is_invariant_under_insertion}
    Suppose that $(w'_1,\dots,w'_k)\in (\Sigma^{n'})^k$ is obtained from $(w_1,\dots,w_k)\in (\Sigma^{n})^k$ by a $\Delta$-deletion for some subset $\Delta \subseteq \Sigma$. 
    Then the following statements are equivalent.
    \begin{enumerate}
        \item $w_1,\dots,w_k$ are $\Delta$-cyclically equalizable.
        \item $w'_1,\dots,w'_k$ are $\Delta$-cyclically equalizable.
    \end{enumerate}
\end{corollary}

\section{Cyclic Equalizability of Two Binary Words}

Throughout this section, we assume that $\Sigma = \{0,1\}$.
The following theorem gives a complete characterization on cyclic equalizability of two binary words.

\begin{theorem}
    \label{thm:criterion_for_two_binary_words}
    Two binary words $w_1,w_2 \in \{0,1\}^n$ are cyclically equalizable if and only if $\wt(w_1) = \wt(w_2)$.
\end{theorem}

\begin{proof}
    Since the only if part is trivial, we show the if part. 
    Suppose that $\wt(w_1) = \wt(w_2)$. From Corollary~\ref{cor:equalizability_is_invariant_under_insertion}, we can assume that no more $\{0,1\}$-deletion can be applied to $w_1$ and $w_2$, i.e., each column in the matrix representation of the pair $(w_1,w_2)$ is either $(10)^{\mathrm{T}}$ or $(01)^{\mathrm{T}}$.
    Furthermore, since a cyclic shift of the columns of the matrix does not affect the claim, we can also assume that $w_1$ and $w_2$ are of the following form:
    \[
        \begin{split}
            w_1 &= 1^{\mu_0} 0^{\mu_1} 1^{\mu_2} 0^{\mu_3} \cdots 1^{\mu_{2N-2}} 0^{\mu_{2N-1}};\\
            w_2 &= 0^{\mu_0} 1^{\mu_1} 0^{\mu_2} 1^{\mu_3} \cdots 0^{\mu_{2N-2}} 1^{\mu_{2N-1}},
        \end{split}
    \]
    where each $\mu_j$ is a positive integer and it holds $\sum_{j=0}^{N-1} \mu_{2j} = \sum_{j=0}^{N-1} \mu_{2j+1}$ from the assumption $\wt(w_1) = \wt(w_2)$.
    We assume $N \geq 2$ since $w_1$ and $w_2$ are already cyclically equal if $N = 1$. 

    We say that a pair of binary words $(w'_1, w'_2)$ is \emph{admissible} if $\wt(w'_1) = \wt(w'_2)$ and they are of the following form:
    \[
        \begin{split}
            w'_1 &= 1^{\nu'_0} 1^{\mu'_0} 0^{\nu'_1} 0^{\mu'_1} 1^{\nu'_2} 1^{\mu'_2} 0^{\nu'_3} 0^{\mu'_3} \cdots 0^{\nu'_{2N-1}} 0^{\mu'_{2N-1}}; \\
            w'_2 &= 1^{\nu'_0} 0^{\mu'_0} 0^{\nu'_1} 1^{\mu'_1} 1^{\nu'_2} 0^{\mu'_2} 0^{\nu'_3} 1^{\mu'_3} \cdots  0^{\nu'_{2N-1}} 1^{\mu'_{2N-1}},
        \end{split}
    \]
    for some parameters $\nu'_i \geq 0$ and $\mu'_i \geq 1$ ($0 \leq i \leq 2N-1$). 
    The pair $(w_1,w_2)$ in the claim is admissible with $\nu'_i = 0$ and $\mu'_i = \mu_i$. 
    For an admissible pair $(w'_1, w'_2)$ and $0 \leq j \leq 2N-1$, we say that the pair is \emph{consistent} up to the $j$-th slot if the following holds for all $0 \leq i \leq j$:
    \[
    \nu'_i + \mu'_i = \mu'_{(i+1) \bmod 2N} + \nu'_{(i+2) \bmod 2N}.
    \]
    For example, the following pair $(w'_1, w'_2)$ is consistent up to the 1th slot if the underlined subsequences are equal:
    \[
        \begin{split}
            w'_1 &= \underline{1^{\nu'_0} 1^{\mu'_0} 0^{\nu'_1} 0^{\mu'_1}} 1^{\nu'_2} 1^{\mu'_2} 0^{\nu'_3} 0^{\mu'_3} \cdots 0^{\nu'_{2N-1}} 0^{\mu'_{2N-1}}; \\
            w'_2 &= 1^{\nu'_0} 0^{\mu'_0} 0^{\nu'_1} \underline{1^{\mu'_1} 1^{\nu'_2} 0^{\mu'_2} 0^{\nu'_3}} 1^{\mu'_3} \cdots  0^{\nu'_{2N-1}} 1^{\mu_{2N-1}}.
        \end{split}    
    \]
    Note that when $w'_1$ and $w'_2$ are consistent up to the $(2N-3)$-th slot, they are cyclically equal from the following reason. 
    If they are consistent up to the $(2N-3)$-th slot, they are also consistent up to the $(2N-2)$-th slot due to $\wt(w'_1) = \wt(w'_2)$, and furthermore, they are also consistent up to the $(2N-1)$-th slot since they are of equal length. 
    
    In the following, $\{0,1\}$-insertions for $w_1$ and $w_2$ are performed recursively to construct two cyclically equal binary words. 

    First, perform a $\{0,1\}$-insertion to be consistent up to the 0th slot. If $w_1$ and $w_2$ are already consistent up to the 0th slot, i.e., $\mu_0 = \mu_1$, then nothing needs to be done. So we assume that $\mu_0 \neq \mu_1$. 
    If $\mu_0 > \mu_1$, then they will be consistent up to the 0th slot by the following $1$-insertion:
    \[
        \begin{split}
            \widetilde{w_1} &= 1^{\mu_0} 0^{\mu_1} \underline{1^{\mu_0-\mu_1}} 1^{\mu_2} 0^{\mu_3} \cdots 1^{\mu_{2N-2}} 0^{\mu_{2N-1}}; \\
            \widetilde{w_2} &= 0^{\mu_0} 1^{\mu_1} \underline{1^{\mu_0-\mu_1}} 0^{\mu_2} 1^{\mu_3} \cdots 0^{\mu_{2N-2}} 1^{\mu_{2N-1}}.
        \end{split}
    \]
    If $\mu_0 < \mu_1$, then it is enough to do the following $1$-insertion:
    \[
        \begin{split}
            \widetilde{w_1} &= \underline{1^{\mu_1-\mu_0}} 1^{\mu_0} 0^{\mu_1} 1^{\mu_2} 0^{\mu_3} \cdots 1^{\mu_{2N-2}} 0^{\mu_{2N-1}}; \\
            \widetilde{w_2} &= \underline{1^{\mu_1-\mu_0}} 0^{\mu_0} 1^{\mu_1} 0^{\mu_2} 1^{\mu_3} \cdots 0^{\mu_{2N-2}} 1^{\mu_{2N-1}}.
        \end{split}
    \]
    In either case, the current pair $(\widetilde{w_1}, \widetilde{w_2})$ is admissible and consistent up to the 0th slot.

    Next, for $1 \leq j \leq 2N - 3$, when the current pair $(w'_1, w'_2)$ is admissible and consistent up to the $(j-1)$-th slot, we perform $\{0,1\}$-insertions to be consistent up to the $j$-th slot.
    Let $k_0 := \nu'_j + \mu'_j - (\mu'_{j+1} + \nu'_{j+2})$ and $k := |k_0|$.
    If $w'_1$ and $w'_2$ are already consistent up to the $j$-th slot, i.e., $k_0 = 0$, then nothing needs to be done. So we can assume $k_0 \neq 0$.
    If $k_0 > 0$, perform a $0$-insertion of $0^k$ just before $0^{\nu'_{j+2}}$ when $j$ is odd, and a $1$-insertion of $1^k$ just before $1^{\nu'_{j+2}}$ when $j$ is even. 
    The following is an example when $j$ is odd:
    \[
        \begin{split}
            \widetilde{w_1} &= \cdots 1^{\mu'_{j-1}} 0^{\nu'_j} 0^{\mu'_j} 1^{\nu'_{j+1}} 1^{\mu'_{j+1}} \underline{0^{k}} 0^{\nu'_{j+2}} 0^{\mu'_{j+2}} 1^{\nu'_{j+3}} \cdots; \\
            \widetilde{w_2} &= \cdots 0^{\mu'_{j-1}} 0^{\nu'_j} 1^{\mu'_j} 1^{\nu'_{j+1}} 0^{\mu'_{j+1}} \underline{0^{k}} 0^{\nu'_{j+2}} 1^{\mu'_{j+2}} 1^{\nu'_{j+3}} \cdots.
        \end{split}
    \]
    If $k_0 < 0$, for each $0 \leq i \leq j$ such that $i \equiv j \pmod{2}$, 
    a $0$-insertion of $0^k$ just before $0^{\nu'_i}$ is performed when $j$ is odd and a $1$-insertion of $1^k$ just before $1^{\nu'_i}$ when $j$ is even. 
    The following is an example when $j$ is odd:
    \[
        \begin{split}
            \widetilde{w_1} &= 1^{\nu'_0} 1^{\mu'_0} \underline{0^{k}} 0^{\nu'_1} 0^{\mu'_1} 1^{\nu'_2} 1^{\mu'_2} \underline{0^{k}} 0^{\nu'_3} 0^{\mu'_3}\cdots 1^{\mu'_{j-1}} \underline{0^{k}} 0^{\nu'_j} \cdots; \\
            \widetilde{w_2} &= 1^{\nu'_0} 0^{\mu'_0} \underline{0^{k}} 0^{\nu'_1} 1^{\mu'_1} 1^{\nu'_2} 0^{\mu'_2} \underline{0^{k}} 0^{\nu'_3} 1^{\mu'_3}\cdots 0^{\mu'_{j-1}} \underline{0^{k}} 0^{\nu'_j} \cdots.
        \end{split}
    \]
    In either case, the pair $(\widetilde{w_1}, \widetilde{w_2})$ is admissible and consistent up to the $j$-th slot. 
    
    By performing the above procedure up to $j = 2N-3$, we can make $w_1, w_2$ cyclically equal. 
    This completes the proof of Theorem \ref{thm:criterion_for_two_binary_words}.
\end{proof}

\section{Application to Card-Based Cryptography}\label{sec:application}

\subsection{Information Erasure}\label{ss:information_erasure}

Suppose that we have a sequence of face-down cards $\vec{x} = (x_1, x_2, \ldots, x_n) \in \Sigma^n$. 
We know that $\vec{x} \in S \subseteq \Sigma^n$ but do not know which one. 
Consider the protocol in the following.

\begin{enumerate}
\item By inserting additional cards, arrange the following sequence:
\[
\underset{x_1}{\back}\,\underbrace{\back\,\back\,\cdots\,\back}_{u_1}\,\underset{x_2}{\back}\,\underbrace{\back\,\back\,\cdots\,\back}_{u_2}\,\underset{x_3}{\back}\,\underbrace{\back\,\back\,\cdots\,\back}_{u_3}\,\cdots\,\underset{x_n}{\back}\,\underbrace{\back\,\back\,\cdots\,\back}_{u_n}\,,
\]
where $u_1, u_2, \ldots, u_n \in \Sigma^*$. 
\item Apply a random cut.
\item Open all cards.
\end{enumerate}

We say that the above protocol is \emph{information erasure} if the probability distribution of the opened symbols does not depend on $\vec{x} \in S$. 
Given $S \subseteq \Sigma^n$, the \emph{information erasure problem} asks whether information erasure is possible, and how to insert if it is possible. 

It is obvious that the information erasure problem can be solved if the cyclic equalizability problem for $k$ words is solved for any $k$. 
From Theorem~\ref{thm:criterion_for_two_binary_words}, the information erasure problem for $S \subseteq \bin^n$ with $|S| = 2$ is completely answered: for $S = \{w_1, w_2\}$, the information erasure problem is possible if and only if $\wt(w_1) = \wt(w_2)$ and how to insert is given in the proof of Theorem~\ref{thm:criterion_for_two_binary_words}. 
The information erasure problems for $S \subseteq \bin^n$ with $|S| \geq 3$ and for $S \subseteq \Sigma^n$ with $|\Sigma|\geq 3$ are open problems. 

\subsection{Single-Cut Full-Open Protocols}\label{ss:scfo}

Suppose that we have a sequence of face-down cards $\vec{x} = (x_1, \ol{x_1}, x_2, \ol{x_2}, \ldots, x_n, \ol{x_n}) \in \bin^{2n}$, i.e., $n$ commitments to $x_1, x_2, \ldots, x_n \in \bin$. 
Let $f\colon \bin^n \ra \bin$ be a function. 
Consider the protocol in the following. 

\begin{enumerate}
\item Permute the order of the sequence by a known permutation $\pi \in S_{2n}$:
\[
\underset{x_1}{\back}\,\underset{\ol{x_1}}{\back}\,
\underset{x_2}{\back}\,\underset{\ol{x_2}}{\back}\,
\cdots\,
\underset{x_n}{\back}\,\underset{\ol{x_n}}{\back}\,
~\rightarrow~
\underset{y_1}{\back}\,\underset{y_2}{\back}\,\underset{y_3}{\back}\,\underset{y_4}{\back}\,\cdots\,\underset{y_{2n}}{\back}\,,
\]
where $(y_1, y_2, \ldots, y_{2n}) = \pi(\vec{x})$. 
\item By inserting additional cards, arrange the following sequence:
\[
\underset{y_1}{\back}\,\underbrace{\back\,\back\,\cdots\,\back}_{u_1}\,\underset{y_2}{\back}\,\underbrace{\back\,\back\,\cdots\,\back}_{u_2}\,\underset{y_3}{\back}\,\underbrace{\back\,\back\,\cdots\,\back}_{u_3}\,\cdots\,\underset{y_{2n}}{\back}\,\underbrace{\back\,\back\,\cdots\,\back}_{u_{2n}}\,,
\]
where $u_1, u_2, \ldots, u_{2n} \in \bin^*$. 
Define $s(\vec{x}) := (y_1, u_1, y_2, u_2, \ldots, y_{2n}, u_{2n})$. 
\item Apply a random cut.
\item Open all cards.
\end{enumerate}

We say that the above protocol is a \emph{single-cut full-open protocol} for $f$ if there exist two binary words $z_0, z_1 \in \bin^*$ such that $z_0, z_1$ are not cyclically equal and two binary words $s(\vec{x})$ and $z_{f(\vec{x})}$ are cyclically equal. 
The number of cards in a protocol is defined as $|s(\vec{x})| = 2n + \sum_{i=1}^{2n}|u_i|$. 

In 1989, den Boer~\cite{BoerEC1989} proposed a single-cut full-open protocol for the AND function $x_1 \wedge x_2$ called the five-card trick (see Section~\ref{ss:combinatorics_from_card}), which is the first card-based protocol in the history. 
In 2014, Heather, Schneider, and Teague~\cite{HeatherFAOC2014} proposed a single-cut full-open protocol for the three-bit equality function, which outputs $1$ if $x_1 = x_2 = x_3$ and $0$ otherwise. 
In 2018, Shinagawa and Mizuki~\cite{ShinagawaICISC2018} rediscovered the Heather--Schneider--Teague's protocol and named it as the six-card trick. 
In addition, they designed a single-cut full-open protocol for the two-bit equality function, which is equivalent to a protocol for $x_1 \oplus x_2$, and posed an open problem of asking whether
there exists a single-cut full-open protocol for an $n$-bit equality function for $n \geq 4$. 
So far, only these three single-cut full-open protocols are known. 
It is an open problem to construct single-cut full-open protocols for other interesting functions. 

Based on the cyclic equalizability, we can derive a lower bound on the number of cards in single-cut full-open protocols. 
For $f\colon \bin^n \ra \bin$ and $b \in \bin$, define $N_b(f) := |\{x \mid f(x) = b\}|$. 
From the definition, we have $N_0(f) + N_1(f) = 2^n$ for any function $f$. 
First, we have a simple lemma as follows. 

\begin{lemma}\label{lemma:cyclic_num}
If $k$ different words $w_1, w_2, \ldots, w_k \in \Sigma^*$ are cyclically equal, then the length of the word $|w_i|$ is at least $k$. 
\end{lemma}

\begin{proof}
Any word $w \in \Sigma^*$ has at most $|w|$ words that are cyclically equal to $w$.
\end{proof}

From Lemma~\ref{lemma:cyclic_num}, we obtain a lower bound on the number of cards. 

\begin{theorem}
The number of cards in a single-cut full-open protocol for a function $f\colon \bin^n \ra \bin$ is at least $\max(N_0(f), N_1(f))$. 
\end{theorem}

\begin{proof}
Let $T$ be the truth table of $f$ whose column corresponds to each symbol in $s(\vec{x})$, i.e., $T$ has $|s(\vec{x})|$ columns and $2^n$ rows. 
From the definition of single-cut full-open protocols, all $N_b(f)$ rows are different binary words and cyclically equal for each $b \in \bin$. 
From Lemma~\ref{lemma:cyclic_num}, the number of cards $|s(\vec{x})|$ is at least $\max(N_0(f), N_1(f))$. 
\end{proof}

\section{Conclusion}

In this paper, we show that two binary words $w_1, w_2 \in \bin^n$ are cyclically equalizable if and only if the Hamming weights of $w_1$ and $w_2$ are equal. 
We also give some applications of cyclic equalizability to the information erasure problem and single-cut full-open protocols. 

Research on cyclic equalizability has just begun and there are many important open problems. 
What is the necessary and sufficient condition for $k$ words $w_1, w_2, \ldots, w_k \in \Sigma^n$ to be cyclically equalizable?
If they are cyclically equalizable, it is also important to minimize the length of the resulting words. 
These general problems are obviously important, but the following more specific problems are also of interest.
\begin{itemize}
\item What is the condition for two words $w_1, w_2 \in \{0,1,2\}^n$ to be cyclically equalizable?
\item Are there two binary words $w_1, w_2 \in \bin^n$ of the same Hamming weight that are not $0$-cyclically equalizable?
\item Are there three binary words $w_1, w_2, w_3 \in \bin^n$ of the same Hamming weight that are not cyclically equalizable? (A counterexample exists for $0$-cyclic equalizability.
Indeed, $1001$, $1010$, and $1100$ are not $0$-cyclically equalizable.)
\item Is there a single-cut full-open protocol for the three-bit AND function? In particular, is there a permutation $\pi \in S_{6}$ such that seven binary words $\{\pi(x_1\ol{x_1}x_2\ol{x_2}x_3\ol{x_3}) \mid (x_1,x_2,x_3) \in \bin^3\setminus \{(1,1,1)\}\}$ are cyclically equalizable?
\item Is there a single-cut full-open protocol for the four-bit equality function? In particular, is there a permutation $\pi \in S_{8}$ such that 14 binary words $\{\pi(x_1\ol{x_1}x_2\ol{x_2}x_3\ol{x_3}x_4\ol{x_4}) \mid (x_1,x_2,x_3,x_4) \in \bin^4\setminus \{(0,0,0), (1,1,1)\}\}$ are cyclically equalizable?
\end{itemize}

\section*{Acknowledgement}

This work was supported by JSPS KAKENHI Grant Numbers 21K17702, 23H00479, and 25K14994, JST CREST Grant Number MJCR22M1, and Institute of Mathematics for Industry, Joint Usage/Research Center in Kyushu University: Short-term Visiting Researcher 2022a006, 2023a009, and 2024a015.

\bibliographystyle{abbrv}
\bibliography{card}

\appendix

\section{Concrete Example of Theorem~\ref{thm:criterion_for_two_binary_words} and Corollary~\ref{cor:equalizability_is_invariant_under_insertion}}

    \begin{enumerate}
    \item Let $w_1, w_2 \in \bin^{20}$ be two binary words of the same Hamming weight as follows:
    \[
        \begin{split}
            w_1 &= 1\dot{0}10\dot{\dot{1}}\ 001\dot{0}\dot{\dot{1}}\ \dot{0}1101\ \dot{0}1000; \\
            w_2 &= 00011\ 11001\ 00010\ 00111.
        \end{split}
    \]
    where they are separated by $5$ bits for ease of reading, and the columns of $(00)^{\mathrm{T}}$ are dotted at the top and the columns of $(11)^{\mathrm{T}}$ are double-dotted at the top.
    \item By applying a $1$-deletion, we can remove the double-dotted columns as follows:
    \[
        \begin{split}
            w'_1 &= 1\dot{0}100\ 01\dot{0}\dot{0}1\ 101\dot{0}1\ 000; \\
            w'_2 &= 00011\ 10000\ 01000\ 111.
        \end{split}
    \]
    \item By applying a $0$-deletion, we can remove the dotted columns as follows:
    \[
        \begin{split}
            w''_1 &= 11000\ 11101\ 1000; \\
            w''_2 &= 00111\ 00010\ 0111.
        \end{split}
    \]
    \item Now we are ready to apply the proof of Theorem \ref{thm:criterion_for_two_binary_words}. By applying a $1$-insertion, we can make it consistent up to the $0$th slot:
    \[
        \begin{split}
            &\underline{\overline{\mathbf{1}}11}00\ 01110\ 11000; \\
            &\mathbf{1}00\underline{11\ 1}0001\ 00111,
        \end{split}
    \]
    where the bold columns are the columns inserted so far (including this time), the overlined columns are the columns inserted this time, and the underlined subsequences are the matching patterns. Hereafter, we use these notations.
    \item We notice that they are already consistent up to $1$th slot since the next sequences just after the underlined sequences are both $000$. 
    The underlines are updated as:
    \[
        \begin{split}
            &\underline{\mathbf{1}1100\ 0}1110\ 11000; \\
            &\mathbf{1}00\underline{11\ 1000}1\ 00111.
        \end{split}
    \]
    \item By applying a $1$-insertion, we can make it consistent up to the $2$th slot:
    \[
        \begin{split}
            &\underline{\mathbf{1}1100\ 0111}0\ \overline{\mathbf{11}}110\ 00;\\
            &\mathbf{1}00\underline{11\ 10001\ \mathbf{11}}001\ 11.
        \end{split}
    \]
    \item By applying a $0$-insertion, we can make it consistent up to the $3$th slot:
    \[
        \begin{split}
            u_1 &= \underline{\mathbf{1}11\overline{\mathbf{0}}0\ 00111\ \overline{\mathbf{0}}0}\mathbf{11}1\ 1000; \\
            u_2 &= \mathbf{1}00\mathbf{0}\underline{1\ 11000\ \mathbf{0}1\mathbf{11}0\ 0}111.
        \end{split}
    \]
    Now $u_1, u_2$ are cyclically equal. Thus we can confirm that $w''_1, w''_2$ are cyclically equalizable. This corresponds to the end of the proof of Theorem \ref{thm:criterion_for_two_binary_words}.
    \item Next we apply Corollary \ref{cor:equalizability_is_invariant_under_insertion}. 
    Let $u'_1, u'_2$ be the binary words obtained by replacing $0$ and $1$ of $w'_1, w'_2$ by $\hat{0} := 01$ and $\hat{1} := 11$, respectively, as follows:
    \[
        \begin{split}
            u'_1 &= \hat{1}\hat{0}\hat{1}\hat{0}\hat{0}\ \hat{0}\hat{1}\hat{0}\hat{0}\hat{1}\ \hat{1}\hat{0}\hat{1}\hat{0}\hat{1}\ \hat{0}\hat{0}\hat{0}; \\
            u'_2 &= \hat{0}\hat{0}\hat{0}\hat{1}\hat{1}\ \hat{1}\hat{0}\hat{0}\hat{0}\hat{0}\ \hat{0}\hat{1}\hat{0}\hat{0}\hat{0}\ \hat{1}\hat{1}\hat{1}.
        \end{split}
    \]
    Then we have $(w_1,w_2) \insto{1} (u'_1, u'_2)$ because $(w_1,w_2)$ does not have two or more consecutive double-dotted columns $(11)^{\mathrm{T}}$.
    \item Let $u''_1, u''_2$ be the binary words obtained by replacing $0$ and $1$ of $w''_1, w''_2$ by $\tilde{0} := \hat{0}\hat{0}\hat{0} = 010101$ and $\tilde{1} := \hat{1}\hat{0}\hat{0} = 110101$, respectively, as follows:
    \[
        \begin{split}
            u''_1 &= \tilde{1}\tilde{1}\tilde{0}\tilde{0}\tilde{0}\ \tilde{1}\tilde{1}\tilde{1}\tilde{0}\tilde{1}\ \tilde{1}\tilde{0}\tilde{0}\tilde{0}; \\
            u''_2 &= \tilde{0}\tilde{0}\tilde{1}\tilde{1}\tilde{1}\ \tilde{0}\tilde{0}\tilde{0}\tilde{1}\tilde{0}\ \tilde{0}\tilde{1}\tilde{1}\tilde{1}.
        \end{split}
    \]
    Now we notice that $u''_1$ and $u''_2$ are obtained by inserting $\hat{0}$'s in $u'_1$ and $u'_2$ because $(u'_1,u'_2)$ has at most two consecutive columns $(\hat{0}\hat{0})^{\mathrm{T}}$. By combining this with $(w_1,w_2) \insto{1} (u'_1, u'_2)$ in Step 8, we have $(w_1,w_2) \xrightarrow{\{0,1\}} (u''_1,u''_2)$. 

    \item Finally, in the same operations to obtain $u_1$ and $u_2$ from $w''_1$ and $w''_2$, the following binary words are obtained by inserting $\tilde{0}$'s and $\tilde{1}$'s in $u''_1$ and $u''_2$ as follows:
    \[
        \begin{split}
            u'''_1&=\mathbf{\tilde{1}}\tilde{1}\tilde{1}\mathbf{\tilde{0}}\tilde{0}\ \tilde{0}\tilde{0}\tilde{1}\tilde{1}\tilde{1}\ \mathbf{\tilde{0}}\tilde{0}\mathbf{\tilde{1}\tilde{1}}\tilde{1}\ \tilde{1}\tilde{0}\tilde{0}\tilde{0}; \\
            u'''_2&=\mathbf{\tilde{1}}\tilde{0}\tilde{0}\mathbf{\tilde{0}}\tilde{1}\ \tilde{1}\tilde{1}\tilde{0}\tilde{0}\tilde{0}\ \mathbf{\tilde{0}}\tilde{1}\mathbf{\tilde{1}\tilde{1}}\tilde{0}\ \tilde{0}\tilde{1}\tilde{1}\tilde{1}.
        \end{split}
    \]
    Now $u'''_1, u'''_2$ are cyclically equal of length $19 \times 6 = 114$. Thus we can confirm that $w_1, w_2$ are cyclically equalizable. 
    \end{enumerate}

The length of the resulting binary words in the above construction depends on the order of the characters to be deleted. In fact, if we apply a $0$-deletion first and then a $1$-deletion, we have the resulting binary words of length $19 \times 4 = 76$. 
\end{document}